\documentclass{amsart}
     
\newtheorem{theorem}{Theorem}[section]

\theoremstyle{definition}

\theoremstyle{remark}
\newtheorem{remark}[theorem]{Remark}
 \usepackage{epsfig} 
 \usepackage{graphics} 
\numberwithin{equation}{section}

\def\a{{\alpha}}

\def\g{{\gamma}}

\def\e{{\epsilon}}

\def\k{{\kappa}}

\def\s{{\sigma}}

\def\f{{{\bf\Phi}}}

\def\m{{\mu}}

\def\p{{{\bf\Psi}}}

\def\ttt{{\theta}}
\def\ee{{\frac{1}{\epsilon}}}
\def\mm{{\frac{1}{\mu}}}
\def\ddd{{{\bf D}}}
\def\bbb{{{\bf B}}}
\def\nnn{{{\bf n}}}
\def\aaa{{{\bf A}}}

\def\nnnn{{{\bf N}}}

\def\G{{\Gamma}}
\def\O{{\Omega}}
\def\D{{\Delta}}
\def\ggg{{{\bf G}}}

\def\bb{{{\mathcal B}}}

\def\dd{{{\mathcal D}}}

\def\kk{{{\mathcal K}}}

\def\tt{{{\mathcal T}}}

\def\ggg{{{\bf G}}}
\def\vvv{{{\bf v}}}
\def\xxx{{{\bf x}}}
\def\nnn{{{\bf n}}}

\def\R{{{\bf R}^1}}
\def\RR{{{\bf R}^2}}
\def\RRR{{{\bf R}^3}}

\def\cr{{{\bf curl}}}

\def\9{{\ \hbox{in}\ \O}}
\def\1{{\ \hbox{on}\ \G_1}}
\def\2{{\ \hbox{on}\ \G_2}}
\def\3{{\ \hbox{on}\ \G_3}}
\def\0{{\ \hbox{on}\ \G}}

\def\pa{{\partial}}
\def\pp{{\parallel}}

\begin{document}

\title[A nonlocal formulation for the problem of microwave heating]{A nonlocal formulation for the problem of microwave heating of material with temperature dependent conductivity}
\author{Giovanni Cimatti}
\address{Department of Mathematics, Largo Bruno
  Pontecorvo 5, 56127 Pisa Italy}
\email{cimatti@dm.unipi.it}






\begin{abstract}
Microwave electromagnetic heating are widely used in many industrial processes. The mathematics involved is based on the Maxwell's equations coupled with the heat equation. The thermal conductivity is strongly dependent on the temperature, itself an unknown of the system of P.D.E. We propose here a model which simplifies this coupling using a nonlocal term as the source of heating. We prove that the corresponding mathematical initial-boundary value problem has solutions using the Schauder's fixed point theorem.
\end{abstract}

\maketitle

\section{Introduction}
Microwave electromagnetic heating is increasingly used in industrial applications. The existing treatments deals mainly with one-dimensional models (see \cite{PS}, \cite{MS}, \cite{HD} and reference therein). In this paper the coupling between the Maxwell's equations and the heat equation is modeled in a simplified form, since the source term in the heat equation is taken to be

\begin{equation*}
E(t)= \frac{1}{2}\int_\O \Bigl[\ee\ddd^2(t)+\mm\bbb^2(t)\Bigl]dx
\end{equation*}
which represents the total electromagnetic energy This is justified by the order of magnitude of the parameters involved, moreover this approach simplifies the mathematical treatment and is probably also useful for the numerical treatment.

 Let $\O$ be an open and bounded subset of $\RRR$ (or of $\RR$) representing a conductor of both heat and electricity. The boundary of $\O$, denoted $\G$, is supposed to be of class $C^2$. The electrical conductivity $\s$ depends from the temperature $\ttt$ and the position, according to a given law: $\s=\s(\ttt,\xxx)$. We wish to determine the electric induction $\ddd(\xxx,t)$, the magnetic induction $\bbb(\xxx,t)$ and the temperature $\ttt(\xxx,t)$ in $\O$ with the following initial boundary-value problem

\begin{equation}
\frac{\pa\ddd}{\pa t}+\s(\ttt,\xxx)\ee\ddd-\mm\ \cr\ \bbb=\ggg\quad \hbox{in}\quad Q_{T}:=\O\times(0,T)\quad 
\label{1_2}
\end{equation}

\begin{equation}
\frac{\pa\bbb}{\pa t}+\ee\ \cr\ \ddd=0\quad\hbox{in}\ Q_{T}
\label{2_2}
\end{equation}

\begin{equation}
\frac{\pa\ttt}{\pa t}-\kappa\D\ttt=E(t) \quad \hbox{in}\quad Q_{T}
\label{3_2}
\end{equation}

\begin{equation}
\ddd(\xxx,0)=\ddd_0(\xxx)\quad \xxx\in\O
\label{4_2}
\end{equation}

\begin{equation}
\bbb(\xxx,0)=\bbb_0(\xxx)\quad \xxx\in\O
\label{5_2}
\end{equation}

\begin{equation}
\ttt(\xxx,0)=\ttt_0(\xxx)\quad \xxx\in\O
\label{6_2}
\end{equation}

\begin{equation}
\ttt=0\quad \hbox{on}\ \G\times(0,T)
\label{7_2}
\end{equation}

\begin{equation}
\nnn\wedge\ddd=0\quad \hbox{on}\ \G\times(0,T),
\label{8_2}
\end{equation}

where $\ddd_0(\xxx)$, $\bbb_0(\xxx)$ and $\theta_0(\xxx)$ are given initial data. The dielectric constant $\e$, the magnetic permeability $\m$ and the diffusivity $\kappa$ are all assumed to be constants. $\nnn$ is the unit vector normal to $\G$. The term in the R.H.S of equation (\ref{1_2}) reflects a possible generation of charge in $\O$ as e.g. for the presence of radio-active material in $\O$. \footnote{This model apply  to micro-wave heating not to be confused with induction heating which is based on the quasi-stationary Maxwell equations.}

In this paper we prove that, under suitable assumptions on the data, the problem (\ref{1_2})-(\ref{8_2}) has solutions. The main analytical tool will be the Schauder's theorem used to make a fixed point on the total electromagnetic energy $E(t)$.

\section{Weak formulation of the problem}

In addition to the usual Lebesgues and Sobolev spaces use will be made of the following spaces (see for more details \cite{DL}, \cite {MC}) and \cite{RDL})

\begin{equation}
H(\cr;\O)=\big\{\p\in L^2(\O)^3,\ \cr\ \p\in L^2(\O)^3\big\}
\label{1_5}
\end{equation} 

\begin{equation}
H_0(\cr;\O)=\big\{\f\in L^2(\O)^3,\ \cr\ \f\in L^2(\O)^3,\ \nnn\wedge\p_{|\G}=0\big\}.
\label{2_5}
\end{equation} 

We recall that the mapping $\vvv\to(\nnn\wedge\vvv)_{\G}$ from $C^1(\bar\O)^3$ into $C^1(\G)^3$ can be extended by continuity to a mapping, again denoted $\vvv\to(\nnn\wedge\vvv)_{\G}$, from $H(\cr;\O)$ into $\big(H^{-1/2}(\G)\big)^3$ . We recall the Green formula \footnote{Here and hereafter $(\aaa,\bbb)$ denotes the scalar product in $L^2(\O)^3$ and $\pp\ \pp$ the corresponding norm. Moreover, we use the notation $\int_\O \aaa(\xxx,t)\bullet\bbb(\xxx,t)dx=(\aaa(t)\bullet\bbb(t))$.}

\begin{equation}
\bigl(\cr\ \aaa,\bbb\bigl)-\bigl(\aaa,\cr\ \bbb\bigl)=\int_\G\nnn\wedge\aaa\bullet\bbb d\ \G
\label{3_5}
\end{equation}
for all $\aaa$ and $\bbb \ \in H(\cr,\O)$. If, in particular, either $\aaa$ or $\bbb\in H_0(\cr;\O)$ then

\begin{equation}
\bigl(\cr\ \aaa,\bbb\bigl)=\bigl(\aaa,\cr\ \bbb\bigl).
\label{4_5}
\end{equation}

We motivate here our weak formulation of problem (\ref{1_2})-(\ref{8_2}) . Let $T>0$ and $(\ddd,\bbb,\ttt)$ be a classical solution. Let

\begin{equation}
\f(\xxx,t)\in L^2(0,T;H_0(\cr;\O)),\ \frac{\pa\f}{\pa t}\in L^2(0,T;L^2(\O)^3),\ \f(\xxx,T)=0.
\label{3_7}
\end{equation}

Multiplying (\ref{1_2}) by $\f$ and integrating by parts we have

\begin{equation}
\int_0^T\Bigl[-\Bigl(\ddd,\frac{\pa\f}{\pa t}\Bigr)+\ee\Bigr(\s(\ttt,\xxx)\ddd,\f\Bigl)-\mm\Bigl(\bbb,\cr\ \f \Bigr)\Bigl]dt=\Bigl(\ddd_0(\xxx),\f(\xxx,0)\Bigl)-\int_0^T\Bigl(\ggg,\f\Bigl)dt,
\label{1_9}
\end{equation}

for all $\f(\xxx,t)$ satisfying (\ref{3_7}). Let  
\begin{equation}
\p(\xxx,t)\in L^2(0,T;H(\cr;\O)),\ \frac{\pa\p}{\pa t}\in L^2(0,T;L^2(\O)^3),\ \p(\xxx,T)=0.
\label{1_10}
\end{equation}

Multiplying (\ref{2_2}) by $\p(\xxx,t)$ and integrating by parts with respect to $t$ we obtain

\begin{equation}
\int_0^T\bigg[\Bigl(-\bbb,\frac{\pa\p}{\pa t}\Bigl)+\ee\Bigl(\ddd,\cr\ \p\Bigl)\Bigl]dt=\Bigl(\bbb_0(\xxx),\p(\xxx,0)\Bigl)
\label{2_10}
\end{equation}
for all $\p(\xxx,t)$ satisfying (\ref{1_10}). Therefore, we arrive to the following problem. Let

\begin{equation}
\ggg\in L^2(0,T;L^3(\O)^3),\quad \frac{\pa\ggg}{\pa t}\in  L^2(0,T;L^3(\O)^3)
\label{a_11}
\end{equation}

\begin{equation}
\ddd_0(\xxx)\in H_0(\cr;\O),\quad \bbb_0(\xxx)\in H(\cr;\O)
\label{b_11}
\end{equation}

\begin{equation}
\ttt_0(\xxx)\in H_0^1(\O).
\label{c_11}
\end{equation}

Assume
\begin{equation}
\s(\xi,\xxx)\in C^1(\bar Q_T),\quad |\s(\xi,\xxx)|\leq \s_0,\quad\Bigl|\frac{\pa\s}{\pa\xi}\bigg|\leq\s_1.
\label{c_111}
\end{equation}

We wish to find 

\begin{equation}
\ddd\in L^\infty(0,T;H_0(\cr;\O)),\quad \frac{\pa\ddd}{\pa t}\in L^\infty(0,T;L^2(\O)^3)
\label{d_11}
\end{equation}

\begin{equation}
\bbb\in L^\infty(0,T;H(\cr;\O)),\quad \frac{\pa\bbb}{\pa t}\in L^\infty(0,T;L^2(\O)^3)
\label{e_11}
\end{equation}

\begin{equation}
\ttt\in L^\infty(0,T;H_0^1(\O))
\label{f_11}
\end{equation}

such that (\ref{1_9}) and (\ref{2_10}) hold and

\begin{equation}
\ttt_t=\k\D\ttt+\frac{1}{2}\int_\O\Bigl[\ee\ddd^2(t)+\mm\bbb^2(t)\Bigl]dx.
\label{h_11}
\end{equation}

\section{The linear problem}
\noindent To prove that problem (\ref{1_9})-(\ref{h_11}) has at least one solution we study in this Section a linear problem for the Maxwell' equations in which the electrical conductivity $s$ is assigned as a function of $\xxx$ and $t$. 

\begin{theorem}
 Let

\begin{equation}
\ggg\in L^2(0,T;L^2(\O)^3),\quad \frac{\pa\ggg}{\pa t}\in  L^2(0,T;L^2(\O)^3)
\label{aa_11}
\end{equation}

\begin{equation}
\ddd_0(\xxx)\in H_0(\cr;\O),\quad \bbb_0(\xxx)\in H(\cr;\O).
\label{bb_11}
\end{equation}

Assume $s(\xxx,t)\in C^1(\bar Q_T)$ to satisfy
\begin{equation}
 |s(\xxx,t)|\leq\ s_0
\label{cc_11}
\end{equation}

\begin{equation}
\Bigl |\frac{\pa s}{\pa t}(\xxx,t)\Bigl|\leq\ s_1.
\label{ccc_11}
\end{equation}

There exists a unique weak solution to the initial-boundary problem

\begin{equation}
\frac{\pa\ddd}{\pa t}+s(\xxx,t)\ee\ddd-\mm\cr\ \bbb=\ggg
\label{5_14}
\end{equation}

\begin{equation}
\frac{\pa\bbb}{\pa t}+\ee\ \cr\ \ddd=0
\label{6_14}
\end{equation}

\begin{equation}
\ddd\in L^\infty(0,T;H_0(\cr;\O)),\quad \frac{\pa\ddd}{\pa t}\in L^\infty(0,T;L^2(\O)^3)
\label{ddd_11}
\end{equation}

\begin{equation}
\bbb\in L^\infty(0,T;H(\cr;\O)),\quad \frac{\pa\bbb}{\pa t}\in L^\infty(0,T;L^2(\O)^3)
\label{eee_11}
\end{equation}

\begin{equation}
\ddd(\xxx,0)=\ddd_0(\xxx)
\label{7_14}
\end{equation}

\begin{equation}
\bbb(\xxx,0)=\bbb_0(\xxx).
\label{8_14}
\end{equation}
\end{theorem}

\begin{proof}\footnote{The semigroup theory, as e.g. in \cite{E}, is not directly applicable since, in the present case, the Maxwell operator depends explicitly on $t$.} \footnote{The proof of this theorem is modeled after theorems 4.1 and 5.1 of \cite{DL}. However, these theorems are no directly applicable to the present case since here $s$ depends on $t$.}

We apply the Galerkin method. Let $\f_j(\xxx)\in C^1(\bar\O)^3$, $j\in\nnnn$ be such that

\begin{equation}
\f_j\wedge\nnn=0\quad\hbox{on}\ \G\quad\hbox{for all}\ j
\label{i_15}
\end{equation}

\begin{equation}
\{\f_1,....,\f_m\}\quad\hbox{are linearly independent for every}\ m.
\label{ii_15}
\end{equation}

\begin{equation}
\hbox{The finite sums}\ \sum_{j=1}^m\xi_j\f_j(\xxx)\quad \xi\in\R\ \hbox{are dense in}\ H_0(\cr;\O).\footnote{The ``base'' $\f_1,\f_2,...\f_m,...$ do exist and the same can be said for $\p_j.$ See \cite{DL} page 147. }
\label{iii_15}
\end{equation}

 Let $\p_j(\xxx)\in C^1(\bar\O)^3$, $j\in\nnnn$ be such that

\begin{equation}
\{\p_1,....,\p_m\}\quad\hbox{are linearly independent for every }\ m. 
\label{iv_15}
\end{equation}

\begin{equation}
\hbox{The finite sums}\ \sum_{j=1}^m\xi_j\p_j(\xxx)\quad \xi\in\R\ \hbox{are dense in}\ H(\cr;\O).
\label{v_15}
\end{equation}

By (\ref{v_15}) and (\ref{iii_15}) we can find two sequences $\ddd_{0m}(\xxx)$, $\bbb_{0m}(\xxx)$ of the form

\begin{equation*}
\ddd_{0m}(\xxx)=\sum_{j=1}^m\xi_j\f_j(\xxx),\quad \bbb_{0m}(\xxx)=\sum_{j=1}^m\eta_j\p_j(\xxx)
\end{equation*}

such that for $m\to\infty$

\begin{equation}
\label {1}
\ddd_{0m}(\xxx)\rightarrow \ddd_0(\xxx)\quad \hbox{in}\ H_0(\cr;\O),\quad \bbb_{0m}(\xxx)\rightarrow \bbb_0(\xxx)\quad \hbox{in}\ H(\cr;\O).
\end{equation}

 We look for approximate solutions of the form

\begin{equation}
\ddd_m(\xxx,t)=\sum_{j=1}^m d_j(t)\f_j(\xxx),\ \bbb_m(\xxx,t)=\sum_{j=1}^m b_j(t)\p_j(\xxx),
\label{1_16}
\end{equation}
where $d_j(t)$, $b_j(t)$ are solutions of the following system of $2m$ ordinary differential equations

\begin{equation}
\label{3_19}
\big(\ddd_m'(t),\f_j\big) =\mm\big(\cr\ \bbb_m(t),\f_j\big)-\ee\big( s(\xxx,t)\ddd_m(t),\f_j\big)+\big(\ggg(t),\f_j\big)
\end{equation}

\begin{equation}
\label{1_20}
\big(\bbb_m'(t),\p_j\big)=-\ee\big(\cr\ \ddd_m(t),\p_j\big)
\end{equation}
supplemented with the initial conditions

\begin{equation}
\label{A}
\ddd_m(0)=\ddd_{0m},\quad \bbb_m(0)=\bbb_{0m}.
\end{equation}

By the theorem of existence and uniqueness for systems of ordinary differential equations  the linear initial value problem (\ref{3_19})- (\ref{A}) has one and only one solution $(d_j(t),b_j(t))$ defined in $[0,T]$.

 Let us multiply (\ref{3_19}) by $d_j(t)$ and (\ref{1_20}) by $b_j(t)$ and sum over $j$ from $1$ to $m$. We find respectively

\begin{equation}
\bigl(\ddd_m'(t),\ddd(t)\bigl)=\mm\bigl(\cr\ \bbb_m(t),\ddd_m(t)\bigl)-\ee\bigl(s(\xxx,t)\ddd_m(t),\ddd_m(t)\bigl)+\bigl(\ggg(t),\ddd_m(t)\bigl)
\label{2_20a}
\end{equation}

and

\begin{equation}
\bigl(\bbb_m'(t),\bbb_m(t)\bigl)=-\ee\bigl(\cr\ \ddd_m(t),\bbb_m(t)\bigl).
\label{3_20a}
\end{equation}

Multiplying (\ref{2_20a}) by $\ee$ and (\ref{3_20a}) by $\mm$, adding the resulting equations  and recalling (\ref{4_5}) we obtain

\begin{equation}
\frac{1}{2}\frac{d}{dt}\Bigl(\ee\pp\ddd_m(t)\pp^2+\mm\pp\bbb_m(t)\pp^2\Bigl)=-\ee\Bigl(s(t)\ddd_m(t),\ddd_m(t)\Bigl)+\Bigl(\ggg(t),\ddd_m(t)\Bigl).
\label{1_21}
\end{equation}

Recalling (\ref{cc_11}) we can estimate the right hand side of (\ref{1_21}) and arrive at 

\begin{equation}
\frac{1}{2}\frac{d}{dt}\Bigl(\ee\pp\ddd_m(t)\pp^2+\mm\pp\bbb_m(t)\pp^2\Bigl)\leq C_1+C_2\Bigl(\ee\pp\ddd_m(t)\pp^2+\mm\pp\bbb_m(t)\pp^2\Bigl),
\label{2_21}
\end{equation}

where the constants $C_1$ and $C_2$ depend only on the data. Moreover, by (\ref{1}) and (\ref{A}) we have

\begin{equation}
\ee\pp\ddd_m(0)\pp^2+\mm\pp\bbb_m(0)\pp^2\leq C,
\label{3_21}
\end{equation}
where the constant $C$ depends only on the data. Using the Gronwall's lemma we conclude that

\begin{equation}
\ddd_m\ \hbox{and}\ \bbb_m\ \hbox{are bounded in}\ L^\infty(0,T;L^2(\O)^3)\ \hbox{by a constant}\  M,
\label{4_21}
\end{equation}
with $M$ not depending on $m$. Therefore, we can extract two subsequences $\ddd_\mu$ and $\bbb_\mu$ such that

\begin{equation}
\ddd_\mu\to\ddd,\quad \bbb_\mu\to\bbb\quad\hbox{weakly* in}\ L^\infty(0,T;L^2(\O)^3).
\label{1_22}
\end{equation}

Let $\xi_j(t)\in C^1([0,T]),\ \xi_j(T)=0,\ \eta_j(t)\in C^1([0,T]),\ \eta_j(T)=0$ and define

\begin{equation}
{\bf V}(\xxx,t)=\sum_{j=1}^m \xi_j(t)\f_j(\xxx)
\label{3_22}
\end{equation}

\begin{equation}
 {\bf W}(\xxx,t)=\sum_{j=1}^m \eta_j(t)\p_j(\xxx).
\label{3_22b}
\end{equation}
Let us choose in (\ref{3_19}) and (\ref{1_20}) $m=\mu$. Multiply (\ref{3_19}) by $\xi_j(t)$ and (\ref{1_20} by $\eta_j(t)$. Summing over $j$ from 1 to $m$, integrating by parts and using  (\ref{4_5}) we obtain

\begin{equation}
\int_0^T-\Bigl(\ddd_\mu,\frac{\pa{\bf V}}{\pa t}\Bigl) dt=\int_0^T\mm\Bigl( \bbb_\mu,\cr\ {\bf V}\Bigl)dt-\int_0^T\ee\Bigl(s(t)\ddd_\mu,{\bf V}\Bigl)dt+
\label{1_23}
\end{equation}

\begin{equation*}
\int_0^T\bigl(\ggg,{\bf V}\bigl)dt+\bigl(\ddd_0(\xxx),{\bf V(\xxx,0)}\bigl)
\end{equation*}

and

\begin{equation}
-\int_0^T\Bigl(\bbb_\mu,\frac{\pa{\bf W}}{\pa t}\Bigl)dt=-\int_0^T\ee\Bigl(\ddd_\mu,\cr\ {\bf W}\Bigl)dt+\Bigl(\bbb_0(\xxx),{\bf W(\xxx,0)}\Bigl).
\label{2_23}
\end{equation}
By (\ref{1_22}) we can pass to the limit in (\ref{1_23}) and (\ref{2_23}) and obtain

\begin{equation}
\int_0^T-\Bigl(\ddd,\frac{\pa{\bf V}}{dt}\Bigl) dt=\int_0^T\mm\Bigl( \bbb,\cr\ {\bf V}\Bigl)dt-\int_0^T\ee\big(s(t)\ddd,{\bf V}\big)dt+
\label{1_24}
\end{equation}

\begin{equation*}
\int_0^T\bigl(\ggg,{\bf V}\bigl)dt+\bigl(\ddd_0(\xxx),{\bf V(\xxx,0)}\bigl)
\end{equation*}

and

\begin{equation}
-\int_0^T\Bigl(\bbb,\frac{\pa{\bf W}}{dt}\Bigl)dt=-\int_0^T\ee\Bigl( \ddd,\cr\ {\bf W}\Bigl)dt+\big(\bbb_0(\xxx),{\bf W(\xxx,0)}\big)
\label{2_24}
\end{equation}
for all functions of the special form (\ref{3_22}) and (\ref{3_22b}). But, if $\f$ satisfies (\ref{3_7}) we can find a sequence ${\bf V}_k(\xxx,t)$ of the form (\ref{3_22}) such that

\begin{equation}
{\bf V}_k(\xxx,t)\to \f(\xxx,t)\quad \hbox{in}\ L^2(0,T;H_0(\cr;\O))
\label{2_25}
\end{equation}
                 
\begin{equation}
 \frac{\pa{\bf V}_k}{\pa t}\to \frac{\pa\f}{\pa t}\quad \hbox{in}\ L^2(0,T;L^2(\O)^3).
\label{3_25}
\end{equation}
Furthermore,  if  $\p$ satisfies (\ref{1_10}), we can find a sequence ${\bf W}_k(\xxx,t)$ of the form (\ref{3_22b}) such that

\begin{equation}
{\bf W}_k(\xxx,t)\to \p(\xxx,t)\quad \hbox{in}\ L^2(0,T;H(\cr;\O))
\label{4_25}
\end{equation}
                 
\begin{equation}
 \frac{\pa{\bf W}_k}{\pa t}\to \frac{\pa\p}{\pa t}\quad \hbox{in}\ L^2(0,T;L^2(\O)^3).
\label{5_25}
\end{equation}
With this choice of the test functions the equation (\ref{1_24}) becomes

\begin{equation}
\int_0^T-\Bigl(\ddd,\frac{\pa{\bf V_k}}{dt}\Bigl) dt=\int_0^T\mm\Bigl( \bbb,\cr\ {\bf V_k}\Bigl)dt-\int_0^T\ee\bigl(s(t)\ddd,{\bf V_k}\bigl)dt+
\label{1_24bis}
\end{equation}

\begin{equation*}
\int_0^T\bigl(\ggg,{\bf V_k}\bigl)dt+\bigl(\ddd_0(\xxx),{\bf V_k(\xxx,0)}\bigl).
\end{equation*}
In the same way, with the choice of the test functions ${\bf W_k}$, (\ref{2_24}) becomes

\begin{equation}
-\int_0^T\Bigl(\bbb,\frac{\pa{\bf W}_k}{dt}\Bigl)dt=\int_0^T\ee\Bigl(\ddd,\cr\ {\bf W}_k\Bigl)dt+\Bigl(\bbb_0(\xxx),{\bf W}_k(\xxx,0)\Bigl).
\label{2_24bis}
\end{equation}

In the limit for $k\to\infty$ from (\ref{1_24bis}) we obtain

\begin{equation}
\int_0^T\Bigl[-\Bigl(\ddd,\frac{\pa \f}{\pa t}\Bigl)-\mm\Bigl(\bbb,\cr\ \f\Bigl)+\ee\Bigl( s(\xxx,t)\ddd,\f\Bigl)\Bigl]dt=
\label{2_26}
\end{equation}

\begin{equation*}
\int_0^T\Bigl(\ggg,\f\Bigl)dt+\Bigl(\ddd_0(\xxx),\f(\xxx,0)\Bigl)
\end{equation*}
for every $\f(\xxx,t)$ satisfying (\ref{3_7}). Moreover, from (\ref{2_24bis}), again for $k\to\infty$, we have

\begin{equation}
\int_0^T\Bigl[-\Bigl(\bbb,\frac{\pa\p}{\pa t}\Bigl)-\ee\Bigl(\ddd,\cr \p\Bigl)\Bigl]dt=\Bigl(\bbb_0(\xxx),\p(\xxx,0)\Bigl)
\label{3_26}
\end{equation}
for every $\p(\xxx,t)$ satisfying (\ref{1_10}). Thus ($\ddd$,$\bbb$) is a weak solution to problem (\ref{5_14})-(\ref{8_14}).

 Up to now we did not use all the assumptions of regularity made on the data. If we do this we can obtain a more regular solution. Let us differentiate with respect to $t$ (\ref{3_19}) and (\ref{1_20}). We find

\begin{equation}
\label{1_28}
\Bigl(\ddd_m''(t),\f_j\Bigl)=\mm\Bigl(\cr\ \bbb_m'(t),\f_j\Bigl)-\ee\Bigl( s(\xxx,t)\ddd_m'(t),\f_j\Bigl)-\ee\Bigl(\frac{\pa s(\xxx,t)}{\pa t}\ddd_m(t),\f_j\Bigl)+\Big(\ggg'(t),\f_j\Bigl)
\end{equation}

\begin{equation}
\label{2_28}
\Bigl(\bbb_m''(t),\p_j\Bigl)=-\ee\Bigl(\cr\ \ddd_m'(t),\p_j\Bigl).
\end{equation}

Multiplying (\ref{1_28}) by $\ee d_j(t)$ and (\ref{2_28}) by $\mm b_j(t)$ and summing over $j$ we obtain respectively

\begin{equation}
\ee\frac{1}{2}\frac{d}{dt}\pp \ddd_m'(t)\pp^2=\ee\mm\Bigl(\cr\ \bbb_m'(t),\ddd_m'(t)\Bigl)-\frac{1}{\e^2}\Bigl(\frac{\pa s(\xxx,t)}{\pa t}\ddd_m'(t),\ddd_m(t)\Bigl)-
\label{1_30}
\end{equation}

\begin{equation*}
\frac{1}{\e^2}\Bigl( s(\xxx,t)\ddd_m'(t),\ddd_m'(t)\Bigl)+\ee\Bigl(\frac{\pa\ggg}{\pa t},\ddd_m'(t)\Bigl)
\end{equation*}

\begin{equation}
\mm\frac{1}{2}\frac{d}{dt}\pp\bbb_m'(t)\pp^2=-\mm\ee\Bigl(\cr\ \ddd_m'(t),\bbb_m'(t)\Bigl).
\label{2_30}
\end{equation}
Adding (\ref{1_30}) and \ref{2_30}) and using (\ref{4_5}) we have

\begin{equation}
\frac{1}{2}\frac{d}{dt}\Bigl(\ee\pp\ddd_m'(t)\pp^2+\mm\pp\bbb_m'(t)\pp^2\Bigl)=-\frac{1}{\e^2}\Bigl( s(\xxx,t)\ddd_m'(t),\ddd'_m(t)\Bigl)-\frac{1}{\e^2}\Bigl(\frac{\pa s(\xxx,t)}{\pa t}\ddd_m(t),\ddd_m'\Bigl)+\ee\Bigl(\frac{\pa\ggg}{\pa t},\ddd_m'(t)\Bigl).
\label{3_30}
\end{equation}

Moreover, from (\ref{3_19}) and \ref{1_20}) we have for $t=0$

\begin{equation}
\label{1_31}
\Bigl(\ddd_m'(0),\f_j \Bigl)=\mm\Big(\cr\ \bbb_m(0),\f_j\Bigl)-\ee\Bigl(s(\xxx,0)\ddd_m(0),\f_j\Bigl)+\big(\ggg(0),\f_j)
\end{equation}

\begin{equation}
\label{2_31}
\Bigl(\bbb_m'(0),\p_j\Bigl)=-\ee\Bigl(\cr\ \ddd_m(0),\p_j\Bigl).
\end{equation}

Multiplying (\ref{1_31}) by $\ee d'_j(0)$ and (\ref{2_31}) by $\mm b_j'(0)$ we have, summing over $j$ from $1$ to $m$,

\begin{equation}
\label{alpha}
\ee\pp\ddd_m'(0)\pp^2=\frac{1}{\e\mu}\Bigl(\cr\ \bbb_m(0),\ddd_m(0)\Bigl)-\frac{1}{\e^2}\Bigl(s(\xxx,0),\ddd_m(0),\ddd_m(0)\Bigl)+\ee\Bigl(\ggg(0),\ddd_m(0)\Bigl)
\end{equation}

\begin{equation}
\label{beta}
\mm\pp \bbb_m'(0)\pp^2=-\frac{1}{\e\mu}\Bigl(\cr\ \ddd_m(0),\bbb_m(0)\Bigl)
\end{equation}

Using (\ref{4_5}) we have, adding (\ref{alpha}) and (\ref{beta}), 

\begin{equation}
\label{gamma}
\ee\pp\ddd_m'(0)\pp^2+\mm\pp\bbb_m'(0)\pp^2=-\frac{1}{\e^2}\Bigl(s(\xxx,0),\ddd_m(0)\Bigl)+\ee\Bigl(\ggg(0),\ddd_m(0)\Bigl)\leq C\Bigl(\frac{s_0}{\e^2}+\ee\pp \ggg(0)\pp\Bigl)\pp\ddd_{0m}\pp.
\end{equation}

Since $\pp\ddd_{0m}\pp$ is bounded by (\ref{1}), we conclude, recalling (\ref{aa_11}) that

\begin{equation}
\label{delta}
\ee\pp\ddd_m'(0)\pp^2+\mm\pp\bbb_m'(0)\pp^2\leq M
\end{equation}

with $M$ not depending on $m$. Using (\ref{cc_11}), (\ref{ccc_11}) and the Cauchy-Schwartz inequality we can estimate the R.H.S of (\ref{3_30}) and arrive to the inequality

\begin{equation}
\frac{d}{dt}\Bigl(\ee\pp\ddd_m'(t)\pp^2+\mm\pp\bbb_m'(t)\pp^2\Bigl)\leq C_1 \Bigl(\ee\pp\ddd_m'(t)\pp^2+\mm\pp\bbb_m'(t)\pp^2\Bigl)+C_2.
\label{4_31}
\end{equation}

Recalling (\ref{delta}) and using the Gronwall's inequality we conclude that $\ddd_m'$ and $\bbb'_m$ are both bounded in $L^\infty(0,T;L^2(\O)^3)$ by a constant $N$ not depending on $m$. Hence the weak solution $(\ddd,\bbb)$, found before, satisfies

\begin{equation}
\frac{\pa\ddd}{\pa t}\in L^\infty(0,T;L^2(\O)^3),\quad \frac{\pa\ddd}{\pa t}\in L^\infty(0,T;L^2(\O)^3).
\label{5_31}
\end{equation}
Thus we are permitted to write

\begin{equation}
\ddd(\xxx,0)=\ddd_0(\xxx),\quad\bbb(\xxx,0)=\bbb_0(\xxx).
\label{6_31}
\end{equation}
Integrating by parts with respect to $t$ in (\ref{2_26}), which is now permissible, we have

\begin{equation}
\int_0^T\Bigl[\Bigl(\frac{\pa\ddd}{\pa t},\f\Bigl)-\mm\Bigl(\cr \ \bbb,\f\Bigl)+\ee\Bigl(s(\xxx,t),\f\Bigl)\Bigl]dt=\int_0^T\Bigl(\ggg,\f\Bigl)dt+\Bigl(D_0(\xxx),\f(\xxx,0)\Bigl)
\label{7_31}
\end{equation}

for every
\begin{equation*}
\f(\xxx,t)\in L^2(0,T;H_0(\cr;\O)),\ \frac{\pa\f}{\pa t}\in L^2(0,T;L^2(\O)^3),\ \f(\xxx,0)=0.
\end{equation*}
In the same way, from (\ref{3_26}) we obtain

\begin{equation}
\int_0^T\Bigl[\Bigl(\frac{\pa\bbb}{\pa t},\p\Bigl)+\ee\Bigl(\cr\ \ddd,\p\Bigl)\Bigl]dt=0
\label{2_33}
\end{equation}
for every
\begin{equation*}
\p(\xxx,t)\in L^2(0,T;H(\cr;\O)),\ \frac{\pa\p}{\pa t}\in L^2(0,T;L^2(\O)^3),\ \p(\xxx,0)=0.
\end{equation*}
Thus (\ref{5_14}) and (\ref{6_14}) are satisfied in the sense of distribution. 
\end{proof}

With the degree of regularity on the solution now at our disposal we can also prove easily that the solution is unique. For, we have

\begin{theorem}
If  we assume in Theorem 3.1
\begin{equation}
\ggg=0,\quad \bbb_0=0,\quad \ddd_0=0
\label{u_1}
\end{equation}
the corresponding solution vanishes identically, i.e. the solution of the problem (\ref{aa_11})-(\ref{8_14}) is unique.
\end{theorem}

\begin{proof}
Let us multiply (\ref{5_14}) by $\ee\ddd$ and (\ref{6_14}) by $\mm\bbb$, we obtain, recalling (\ref{4_5}) and taking into account of (\ref{u_1}),

\begin{equation}
\frac{1}{2}\frac{d}{dt}\Bigl(\ee\pp \ddd(t)\pp^2+ \mm\pp\bbb(t)\pp^2\Bigl)=-\ee\Bigl(s(\xxx,t)\ddd(t),\ddd(t)\Bigl).
\label{u_2}
\end{equation}
Estimating the R.H.S we find 

\begin{equation}
\frac{1}{2}\frac{d}{dt}\Bigl(\ee\pp \ddd(t)\pp^2+ \mm\pp\bbb(t)\pp^2\Bigl)\leq C_3\Bigl(\ee\pp \ddd(t)\pp^2+ \mm\pp\bbb(t)\pp^2\Bigl) +C_4.
\label{u_3}
\end{equation}

Since $\ee\pp \ddd(0)\pp^2+ \mm\pp\bbb(0)\pp^2=0$, we conclude, by the Gronwall's Lemma, that $\ddd=\bbb=0$.
\end{proof}

\begin{remark}
If $s(\xxx,t)$ represents an electrical conductivity the physically correct assumption would be $s(\xxx,t)>0$ (see \cite{LL} page 129). However, mathematically this assumption is not needed, in fact the conditions (\ref{cc_11}) and (\ref{ccc_11}) are all what is needed.
\end{remark}

\section{The non-linear problem}

Use will be made in this section of the following ``a priori'' estimate for the heat equation (see \cite{LSU}). If $f(t)\in L^\infty(0,T)$ and $\ttt_0(\xxx)\in H^1_0(\O)$ the weak solution of the following initial-boundary value problem 

\begin{equation}
\ttt_t=\kappa\D\ttt+f(t)\ \hbox{in}\ \O,\ \ttt=0\ \hbox{on}\ \G\times\O,\ \ttt(\xxx,0)=\ttt_0(\xxx) 
\label{3_38}
\end{equation}
belongs to $ H^{1,\infty}(0,T;H^1_0(\O))$ and the following estimate holds

\begin{equation}
\pp\ttt\pp_{H^{1,\infty}(0,T;H^1_0(\O))}\leq K\Bigl(\pp f(t)\pp_{L^\infty(0,T)}+\pp\ttt_0\pp_{ H^1_0(\O)}\Bigl).
\label{2_39}
\end{equation}

 \begin{theorem}
Assume

\begin{equation}
\ggg\in L^2(0,T;L^2(\O)^3),\quad \frac{\pa\ggg}{\pa t}\in  L^2(0,T;L^2(\O)^3)
\label{naa_11}
\end{equation}

\begin{equation}
\ddd_0\in H_0(\cr;\O),\quad \bbb\in H(\cr;\O)
\label{nbb_11}
\end{equation}

\begin{equation} 
\s(\xi,t)\in C^1(Q_T),\quad  |\s(\xi,t)|\leq\s_0,\quad \Bigl\vert\frac{\pa\s(\xi,\xxx)}{\pa \xi}\Bigl\vert\leq\s_1
\label{3_35}
\end{equation}

\begin{equation}
\ttt_0(\xxx)\in H_0^1(\O).
\label{2_35}
\end{equation}

There exists at least one solution of the initial-boundary value problem

\begin{equation}
\ddd\in L^\infty(0,T;H_0(\cr;\O)),\quad \frac{\pa\ddd}{\pa t}\in L^\infty(0,T;L^2(\O)^3)
\label{1_36}
\end{equation}

\begin{equation}
\ddd(\xxx,0)=\ddd_0(\xxx),\quad \bbb(\xxx,0)=\bbb_0(\xxx),\quad \ttt(\xxx,0)=\ttt_0(\xxx)
\label{1_36n}
\end{equation}

\begin{equation}
\bbb\in L^\infty(0,T;H(\cr;\O)),\quad \frac{\pa\bbb}{\pa t}\in L^\infty(0,T;L^2(\O)^3)
\label{3_36}
\end{equation}

\begin{equation}
\ttt\in L^\infty(0,T;H^1_0(\O))
\label{4_36}
\end{equation}

\begin{equation}
\frac{\pa\ddd}{\pa t}-\mm\cr\ \bbb+\s(\ttt,\xxx)\ee\ddd=\ggg
\label{5_36}
\end{equation}

\begin{equation}
\frac{\pa\bbb}{\pa t}+\ee\cr\ \ddd=0
\label{6_36}
\end{equation}

\begin{equation}
\ttt_t=\kappa\D\ttt+\frac{1}{2}\int_\O\Bigl[\ee\ddd^2(t)+\mm\bbb^2(t)\Bigl]dx.
\label{7_36}
\end{equation}
\end{theorem}

\begin{proof}
We obtain easily an ``a priori'' estimate. Assume $(\ddd,\bbb,\ttt)$ to be a solution of problem (\ref{1_36})-(\ref{7_36}). Let us multiply (\ref{5_36}) by $\ee\ddd$ and (\ref{6_36}) by $\mm\bbb$. Integrating over $\O$ and using (\ref{4_5}), we obtain

\begin{equation}
\frac{1}{2}\frac{d}{dt}\Bigl[\ee\pp \ddd(t)\pp^2+\mm\pp\bbb(t)\pp^2\Bigl]=-\int_\O\s(\ttt,\xxx)\ee\ddd^2(t)dx+\ee\Bigl(\ggg(t),\ddd(t)\Bigl).
\label{1_37}
\end{equation}
Estimating the R.H.S of (\ref{1_37}) with the help of (\ref{3_35}), taking into account of the initial conditions (\ref{1_36n}) and using the Gronwall' s lemma, we conclude that there exits a constant $N$, depending only on the data, such that

\begin{equation}
E(t)=\int_\O\Bigl[\ee\ddd(t)^2+\mm\bbb(t)^2\Bigl]dx\leq N. \footnote{We recall that E(t) gives the total electromagnetic energy in $\O$.}
\label{2_37}
\end{equation}
Define $\kk=\Bigl\{f(t)\in C^0([0,T]),\ |f(t)|\leq N\Bigl\}$. Let $E(t)\in\kk$ and solve the problem

\begin{equation}
\hat\ttt_t=\kappa\D\hat\ttt+E(t),\ \hat\ttt(\xxx,0)=\ttt_0(\xxx),\ \hat\ttt=0\quad\hbox{on}\ \G\times(0,T).
\label{1_40}
\end{equation}
We apply Theorem 3.1 with $\hat s(\xxx,t)=\s(\hat\ttt(\xxx,t),\xxx)$. The condition (\ref{cc_11}) is certainly verified by (\ref{3_35}). Moreover, by (\ref{2_39}) and (\ref{ccc_11}) we have

\begin{equation}
\Bigl|\frac{\pa\hat s}{\pa t}(\xxx,t)\Bigl|\leq\Bigl|\frac{\pa\s}{\pa\xi}(\hat\ttt(\xxx,t),\xxx)\Bigl|\Bigl|\frac{\pa\hat\ttt}{\pa t}(\xxx,t)\Bigl|\leq \s_1N.
\label{5_40}
\end{equation}
This implies existence and uniqueness for the following problem

\begin{equation}
\hat\ddd\in L^\infty(0,T;H_0(\cr;\O)),\quad \frac{\pa\hat\ddd}{\pa t}\in L^\infty(0,T;L^2(\O)^3)
\label{nddd_11}
\end{equation}

\begin{equation}
\hat\bbb\in L^\infty(0,T;H(\cr;\O)),\quad \frac{\pa\hat\bbb}{\pa t}\in L^\infty(0,T;L^2(\O)^3)
\label{neee_11}
\end{equation}

\begin{equation}
\hat\ddd(\xxx,0)=\ddd_0(\xxx)
\label{n7_14}
\end{equation}

\begin{equation}
\hat\bbb(\xxx,0)=\bbb_0(\xxx)
\label{n8_14}
\end{equation}

\begin{equation}
\frac{\pa\hat\ddd}{\pa t}-\mm\cr\ \hat\bbb+\hat s(\xxx,t)\ee\hat\ddd=\ggg
\label{1_41}
\end{equation}

\begin{equation}
\frac{\pa\hat\bbb}{\pa t}+\ee\cr\ \hat\ddd=0
\label{2_41}
\end{equation}

\begin{equation}
\hat\ddd(\xxx,0)=\ddd_0(\xxx).
\label{3_41}
\end{equation}

Let

\begin{equation}
\hat E(t)=\frac{1}{2}\int_\O\Bigl[\ee\hat\ddd^2(t)+\mm\hat\bbb^2(t)\Bigl]dx
\label{5_41}
\end{equation}
and define the operator $ \tt:\kk\to L^\infty(0,T)$ by

\begin{equation}
  \hat E(t)=\tt(E(t)).
\label{6_41}
\end{equation}
We claim that

\begin{equation}
\hat E(t)\in H^{1,\infty}(0,T),\quad \pp\hat E(t)\pp_{C^0([0,T])}\leq N.
\label{7_41}
\end{equation}
Differentiating (\ref{5_41}) we have, by (\ref{1_41}) and (\ref{2_41})

\begin{equation}
\hat E'(t)=\ee\mm\Bigl[\Bigl(\hat\ddd(t),\cr\ \hat\bbb(t)\Bigl)-\Bigl(\hat\bbb(t),\cr\ \hat\ddd(t)\Bigl)\Bigl]-\int_\O\ee\s(\hat\ttt(\xxx,t),\xxx)\hat\ddd(t)\bullet\hat\ddd(t) dx+\ee\Bigl(\hat\ddd(t),\ggg(t)\Bigl).
\label{1_42}
\end{equation}
The first integral in R.H.S vanishes by (\ref{4_5}). Thus

\begin{equation}
\hat E'(t)=-\int_\O\ee\s(\hat\ttt(\xxx,t),\xxx)\hat\ddd(t)\bullet\hat\ddd(t) dx+\ee\Bigl(\hat\ddd(t),\ggg(t)\Bigl).
\label{2_42}
\end{equation}
Using (\ref{3_35}) and (\ref{4_21}) we can estimate the R.H.S. of (\ref{2_42}) and obtain

\begin{equation}
\pp\hat E'(t)\pp_{L^\infty (0,T)}\leq C,
\label{3_42}
\end{equation}
where the constant $C$ depends only on the data. Thus (\ref{7_41}) holds true and $\hat E(t)\in \kk$. Since $H^{1,\infty}(0,T)$ is compactly imbedded in $C^0([0,T])$ we concluded that $\tt$ is a compact operator. To apply the Schauder's fixed point theorem it remains to prove that $\tt$ is continuous. 

Assume $c_j(t)\in\kk$ and
\begin{equation}
 c_j(t)\to\bar c(t)\quad \hbox{in}\quad C^0([0,T]).
\label{1_43}
\end{equation}

Let us define
\begin{equation}
\g_j(t)=\tt(c_j(t)),\quad \bar\g(t)=\tt(\bar c(t)).
\label{2_43}
\end{equation}

We claim that $\g_j(t)\to\bar\g(t)$. Let $\bar\ttt$ and $\ttt_j$ be respectively the solutions of the problems

\begin{equation}
\frac{\pa\bar\ttt}{\pa t}=\kappa\D\bar\ttt+\bar c(t),\quad\bar\ttt(\xxx,0)=\ttt_0(\xxx),\quad \bar\ttt=0\ \hbox{on}\ \G\times(0,T)
\label{3_43}
\end{equation}

\begin{equation}
\frac{\pa\ttt_j}{\pa t}=\kappa\D\ttt_j+ c_j(t),\quad\ttt_j(\xxx,0)=\ttt_0(\xxx),\quad \ttt_j=0\ \hbox{on}\ \G\times(0,T).
\label{4_43}
\end{equation}

By difference from (\ref{3_43}) and (\ref{4_43}), recalling (\ref{2_39}) we have

\begin{equation}
\pp\ttt_j-\bar\ttt\pp_{ H^{1,\infty}(0,T;H^1_0(\O))}\leq C\pp c_j(t)-\bar c(t)\pp_{L^\infty(0,T)}.
\label{1_44}
\end{equation}

Set $\s_j=\s(\ttt_j,\xxx)$ and $\bar\s=\s(\bar\ttt,\xxx)$. By (\ref{3_35}) we obtain

\begin{equation}
\bigl|\s_j-\bar\s\bigl|\leq\s_1\pp\ttt_j-\bar\ttt\pp_{ H^{1,\infty}(0,T;H^1_0(\O))}\leq\s_1 C\pp c_j(t)-\bar c(t)\pp_{L^\infty(0,T)}.
\label{1_44_1}
\end{equation}
Let $(\ddd_j,\bbb_j)$ be the solution of problem

\begin{equation}
\frac{\pa\ddd_j}{\pa t}+\ee\s_j\ddd_j-\mm\cr\ \bbb_j=\ggg
\label{4_44}
\end{equation}

\begin{equation}
\frac{\pa\bbb_j}{\pa t}+\ee \cr\ \ddd_j=0
\label{5_44}
\end{equation}

\begin{equation}
\ddd_j(\xxx,0)=\ddd_0,\quad \bbb_j(\xxx,0)=\bbb_0
\label{7_44}
\end{equation}

and $(\bar \ddd,\bar\bbb)$ the solution of problem

\begin{equation}
\frac{\pa\bar \ddd}{\pa t}+\ee\bar\s\bar\ddd-\mm\cr\ \bar\bbb=\ggg
\label{1_45}
\end{equation}

\begin{equation}
\frac{\pa\bar\bbb}{\pa t}+\ee \cr\ \bar\ddd=0
\label{2_45}
\end{equation}

\begin{equation}
\bar\ddd(\xxx,0)=\ddd,\quad \bar\bbb(\xxx,0)=\bbb_0.
\label{4_45}
\end{equation}

Multiplying (\ref{4_44}) by $\ee\ddd_j$ and (\ref{5_44}) by $\mm\bbb_j$, integrating over $\O$, using (\ref{4_5}) and the Gronwall's inequality we conclude that

\begin{equation}
\pp\ddd_j(t)\pp_{L^\infty(0,T;L^2(\O)^3)}\leq A,\quad \pp\bbb_j(t)\pp_{L^\infty(0,T;L^2(\O)^3)}\leq A,
\label{6_45}
\end{equation}
where the constant $A$ depends only on the data. Let us define $\dd_j=\ddd_j-\bar\ddd$,  $\bb_j=\bbb_j-\bar\bbb$. By difference from (\ref{4_44}), (\ref{1_45}) and (\ref{5_44}), (\ref{2_45}) we have

\begin{equation}
\frac{\pa\dd_j}{\pa t}+\ee\bar\s\dd_j+\bigl(\s_j-\bar\s\bigl)\bar\ddd-\mm\cr\ \ \bb_j=0
\label{1_46}
\end{equation}

\begin{equation}
\frac{\pa\bb_j}{\pa t}+\ee\ \cr\ \dd_j=0.
\label{2_46}
\end{equation}

Proceeding as usual, i.e. multiplying (\ref{1_46}) by $\ee\dd_j$ and (\ref{2_46}) by $\mm\bb_j$, integrating over $\O$ and using (\ref{4_5}) we arrive at

\begin{equation}
\frac{1}{2}\frac{d}{dt}\Bigl(\ee\pp\dd_j\pp^2+\mm\pp\bb_j\pp^2\Bigl)=\ee\int_\O\bar\s\dd_j^2dx-\ee\int_\O\bigl(\s_j-\bar\s\bigl)\bar\ddd\bullet\dd_j dx.
\label{3_46}
\end{equation}
Define

\begin{equation}
y_j(t)=\frac{1}{2}\Bigl(\ee\pp\dd_j\pp^2+\mm\pp\bb_j\pp^2\Bigl),\quad \a_j(t)=\pp c_j(t)-\bar c(t)\pp_{C^0([0,T])}.
\label{2_47}
\end{equation}
From (\ref{2_47}) and (\ref{1_44_1}), estimating the R.H.S of (\ref{3_46}) and recalling (\ref{6_45}) we have the differential inequality

\begin{equation}
y_j'(t)\leq K\Bigl[(1+\a_j(t))y_j(t)+\a_j(t)\Bigl],\quad y_j(0)=0,\ K>0.
\label{3_47}
\end{equation}
We use Perron's theorem \footnote{If $y'\leq g(t,y)$ and $z'=g(t,z)$ with $y(0)=z(0)$ then $y(t)\leq z(t)$.} with the auxiliary problem

\begin{equation}
z_j'(t)= K\Bigl[(1+\a_j(t))z_j(t)+\a_j(t)\Bigl],\quad z_j(0)=0,\ K>0.
\label{1_48}
\end{equation}
Hence $y_j(t)\leq z_j(t)$. Since $\a_j(t)\to0$ for $j\to\infty$, we have 

\begin{equation}
z_j(t)\to w(t)\quad \hbox{as}\ j\to\infty,
\label{1_50}
\end{equation}
where $w(t)=0$ is the solution of

\begin{equation}
w'=Kw,\quad w(0)=0.
\label{2_50}
\end{equation}
 Hence $z_j(t)\to 0$ and $y_j(t)\to 0$. This implies the continuity of the operator $\tt$. By the Schauder's fixed point theorem, problem (\ref{1_36})-(\ref{7_36}) has a solution.
\end{proof}

\vskip.5cm
\noindent {\bf Remark.} If we assume $\s=\s_0>0$ and $\ggg=0$, $\s_0$ a constant, i.e. the case physically more common, the derivative of the total energy $E(t)$ is strictly negative by (\ref{2_42}). One could expect the temperature to tend to zero. This is not always the case as the model adopted does not admit dissipation for the magnetic part of the energy. Consider the following  two-dimensional example. Let $\O=\{(x,y);1<x^2+y^2<2\}$, $\ddd_0(x,y)=0$, $\bbb_0(x,y)=\frac{y}{x^2+y^2}{\bf i}-\frac{x}{x^2+y^2}{\bf j}\in C^\infty(\bar\O)$, $\ttt_0(x,y)=\frac{\pi\log\sqrt{x^2+y^2}}{2\kappa\log 2}-\frac{\pi\sqrt{x^2+y^2}}{2D\mu}+\frac{\pi}{2\kappa\mu}$. By direct computation we find as solution to problem (\ref{1_36})-(\ref{7_36})

\begin{equation*}
\ddd(x,y,t)=0,\ \bbb(x,y,t)=\bbb_0(x,y),\ \ttt(x,y,t)=\frac{\pi\log\sqrt{x^2+y^2}}{2\kappa\log 2}-\frac{\pi\sqrt{x^2+y^2}}{2\kappa\mu}+\frac{\pi}{2\kappa\mu}.
\end{equation*}
The temperature does not tend to zero. Essential in this example is the fact that $\O$ is not simply connected.

\bibliographystyle{amsplain}

\begin{thebibliography}{10}

\bibitem{E}M.Eller, Stability of the anisotropic Maxwell equations with a conductivity term, Evol. Eq. and Control Theory, {\bf 8}, (2019), 343-357.

\bibitem{LSU} O.A.Ladyzenskaja, V.A. Solonnikov, N.N. Ural'ceva, Linear and Quasilinear Equations of Parabolic Type, Amer. Math. Soc. Trans. Vol.23, 1968.

\bibitem{RDL}R.Dautray, J.L.Lions, Mathematical Analysis and Numerical Methods for Science and Technology, Springer-Verlag, Vol.3, 1984.

\bibitem{DL}G.Duvaut and J.L.Lions, Inequalities in Mechanics and Physics, Springer-Verlag Berlin, 1976.

\bibitem{MC} M. Cessenat, Mathematical Methods in Electromagnetism, World Scientific, Vol 41, 1996.

\bibitem{LL}L. Landau, E. Lifchitz, Electrodynamique des Milieux Continus, Editions Mir, Moscou, 1969.

\bibitem{PS}A.H. Pincombe, N.F.Smyth, Microwave heating of material with power law temperature dependencies, IMA J. App. Math., {\bf 52}, (1994), 141-176.

\bibitem{HMY}Hong-Ming Yin, Regularity of weak solution to Maxwell's equations and applications to microwave heating, J. Diff. Equa., {\bf 200}, (2004), 137-161.

\bibitem{MS}T.R. Marchant, N.F. Smyth, Microwave heating of material with nonohmic conductance, SIAM J. Appl. Math., {\bf 53}, (1993), 1591-1612.

\bibitem{HD}M.R. Hossan, P. Dutta, Effects of temperature dependence properties in electromagnetic heating, Int. J. Heat and Mass Transfer {\bf 55}, (2012), 3412-3422.
\end{thebibliography}

\end{document}